\documentclass[%
 reprint,
 amsmath,amssymb,
 aps,
]{revtex4-2}

\usepackage{graphicx}
\usepackage{dcolumn}
\usepackage{bm}
\usepackage[margin=1in]{geometry} 
\usepackage{amsmath,amsthm,amssymb,graphicx,mathtools,tikz,hyperref}
\usepackage{braket}
\usepackage{enumitem}
\usepackage{cancel}
\usepackage{amsthm}
\usepackage{xcolor}

\usepackage[export]{adjustbox}
\usetikzlibrary{positioning}

\def\c{\mathbf{c}}
\def\s{\mathbf{s}}

\def\u{\mathbf{u}}
\def\v{\mathbf{v}}
\def\w{\mathbf{w}}
\def\x{\mathbf{x}}

\def\z{\mathbf{z}}

\def\B{\mathbb{B}}
\def\C{\mathbb{C}}

\def\U{\mathbb{U}}

\def\bnz{\mathbb{B}_n^Z}
\def\unz{\mathbb{U}_n^Z}

\def\bOz{\mathbb{B}_1^Z}

\theoremstyle{definition}
\newtheorem{definition}{Definition}[section]
\newtheorem{theorem}{Theorem}[section]
\newtheorem{corollary}{Corollary}[theorem]
\newtheorem{remark}{Remark}[section]

\newtheorem{lemma}[theorem]{Lemma}

\newtheorem{proposition}[theorem]{Proposition}

\hypersetup{
colorlinks,
linkcolor=blue
}

\begin{document}

\preprint{APS/123-QED}

\title{On Quantum Computation Using Bias-Preserving Gates}
\author{Debadrito Roy}
 \affiliation{University of Southern California}
\author{Aryaman Manish Kolhe}%
\affiliation{International Institute of Information Technology, Hyderabad}
\author{V. Lalitha}
\affiliation{International Institute of Information Technology, Hyderabad}
\author{Navin Kashyap}
\affiliation{Indian Institute of Science, Bangalore}


\begin{abstract}
 Certain types of quantum computing platforms, such as those realized using Rydberg atoms or Kerr-cat qubits, are natively more susceptible to Pauli-$Z$ noise than Pauli-$X$ noise, or vice versa. On such hardware, it is useful to ensure that computations use only gates that maintain the $Z$-bias (or $X$-bias) in the noise. This is so that quantum error-correcting codes tailored for biased-noise models can be used to provide fault-tolerance on these platforms. In this paper, we follow up on the recent work of Fellous-Asiani et al.\ (\emph{npj Quantum Inf.,} 2025)  in studying the structure and properties of bias-preserving gates. Our main contributions are threefold: (1)~We give a novel characterization of $Z$-bias-preserving gates based on their decomposition as a linear combination of Pauli operators. (2)~We show that any $Z$-bias-preserving gate can be approximated arbitrarily well using only gates from the set $\{X,R_z(\theta),CNOT,CCNOT\}$, where $\theta$ is any irrational multiple of $2\pi$. (3)~We prove, by drawing a connection with coherence resource theory, that any $Z$-bias-preserving logical operator acting on the logical qubits of a Calderbank-Shor-Steane (CSS) code can be realized by applying $Z$-bias-preserving gates on the physical qubits. Along the way, we also demonstrate that $Z$-bias-preserving gates are far from being universal for quantum computation.
 \end{abstract}

\maketitle

\section{Introduction}

Standard error correction techniques assume single qubit depolarizing noise as the noise model for quantum systems, where $X, Y, Z$ errors occur with equal probability \cite{webster2015reducing}. Even in the case of noise channels other than depolarizing channels, we can adopt the Pauli twirling approximation to convert arbitrary noise channels to Pauli channels \cite{gutierrez2013approximation, geller2013efficient, cai2019constructing}, after which it will be sufficient to correct for $X$ and $Z$ type errors. However, in practice, many physical realizations, for example those realized using superconducting qubits, Rydberg atoms and Kerr-cat qubits, involve a bias towards either $X$ or $Z$ errors \cite{aliferis2009fault, cong2022hardware, puri2020bias}.

Biased error models are described in many ways in present literature depending on the requirements of the problems they try to solve. One approach is through stochastic noise with separate noise strengths for the $X$ and $Z$ errors where noise bias is defined as the ratio of these two values \cite{aliferis2008fault}. One can also model them by assuming that a certain unitary $U$ is applied without error, and is followed by Pauli noise $\Lambda_P$ \cite{martinez2025leveraging} which is a function of the probabilities with which $X$ and $Z$ errors occur. We follow the noise model used in \cite{fellous2025scalable} and assume \emph{perfect bias}, in which only $Z$ errors occur during computation in the case of $Z$-biased error. This is modeled using the Pauli channel with the only non-trivial errors being terms involving $Z$, that is, the channel is a linear combination of $I$ and $Z$ unitaries. This is physically motivated from the idea that it is possible to suppress bit flip errors exponentially in a realization of cat qubits with a linear overhead (in the size of the cat state) of phase flips being introduced \cite{mirrahimi2014dynamically, lescanne2020exponential}. 

 By tailoring the error correcting code to the biased error model, XZZX surface codes \cite{bonilla2021xzzx} and bias-tailored LDPC codes \cite{roffe2023bias} achieved better thresholds than surface codes. If fault-tolerant quantum computation is performed using gates which preserve the noise bias, then error correction becomes efficient in the circuit. Hence, it is of interest to study bias-preserving gates for a given biased error model. Bias-preserving gates have been studied in \cite{aliferis2009fault, guillaud2019repetition, fellous2025scalable}. This part of the literature will be dealt with in detail in the next section (Section \ref{sec:preliminaries}).

\section{Definitions and Background}\label{sec:preliminaries}

 Given that some physical systems are susceptible to biased noise, it is important to ensure that applying gates during quantum computation does not change this noise bias. Bias-preserving gates are quantum gates that preserve the noise bias of a biased error channel. For example, if we consider a $Z$-biased error model, and a $Z$ error occurs, then applying the Hadamard gate converts this error to a bit flip error because $HZ = XH$. As a result, Hadamard gate does not preserve the $Z$-biased noise. 

The notion of bias-preserving gates was first introduced in \cite{aliferis2008fault}, where the authors considered a $Z$ biased error model. They defined bias-preserving gates to be gates that are diagonal in the $Z$ basis as they commute with $Z$ errors applied on any qubit and preserve $Z$ bias even if the error occurs continuously during the application of the gate. We call these continuously-bias-preserving gates. By this definition, $CNOT$ is not a continuously-bias-preserving gate as $Z$ errors that occur during the application of the Hamiltonian that implements $CNOT$ can propagate as $X$ errors. In \cite{guillaud2019repetition}, the authors derive that regardless of the Hamiltonian used to implement $CNOT$, it cannot be made continuously-bias-preserving for finite dimensional systems.

However, recent studies like \cite{puri2020bias} show that with the use of cat qubit states, $CNOT$ can be implemented bias-preservingly.  This motivated the study of a more general class of bias-preserving gates in \cite{fellous2025scalable}. In \cite{fellous2025scalable}, $X$-bias-preserving gates are defined as gates for which $X$ errors upstream of the gate propagate as linear combinations of $X$ errors downstream from the gate. In \cite{tsutsui2024bias}, bias-preserving implementations of $S,H, CZ$ and rotation gates are studied using a repetition code for fault-tolerant computation. 

In this paper, we adhere to the definition of bias-preserving gates given by Fellous-Asiani et al.~\cite{fellous2023scalable}. However, we have chosen to take the $Z$-bias-preserving perspective, instead of the $X$-bias-preserving one in \cite{fellous2023scalable}, mainly because $Z$-bias-preserving gates are essentially permutations of computational basis states --- see Theorem~\ref{thm:bnz-permutation-fellous}. In any case, statements about $Z$-bias-preserving gates can be readily translated into corresponding statements about $X$-bias-preserving gates, as one set of gates is obtained from the other by conjugating with Hadamard gates.
 
 Let $\mathbb{U}_{n}$ denote the (multiplicative) group of unitary operators on $n$ qubits (i.e., $n$-qubit gates). We define the $Z$-restricted Pauli group and $Z$-type unitary operators, similar to their $X$-type counterparts in \cite{fellous2025scalable}. To this end, for a binary vector $\c = (c_1,\ldots,c_n) \in \{0,1\}^n$, we define $Z_\c = \bigotimes_{i=1}^n {Z_i}^{c_i}$, where $Z_i$ denotes the $Z$-gate acting on the $i$-th qubit. Thus, $Z_{\c}$ is the tensor product of $Z$-gates acting on the support of $\c$. 

\begin{definition} The $Z$-restricted Pauli group on $n$ qubits is 
$$\mathcal{P}_n^Z := \{Z_{\c}: \c \in \{0,1\}^n\}.$$
\end{definition}  

\begin{definition} The set of $Z$-type unitary operators on $n$ qubits is
$$\mathbb{U}_n^Z := \bigg\{U \in \mathbb{U}_n \;\big|\; U = \sum_i \alpha_iP_i \text{ for } P_i \in \mathcal{P}_n^Z, \alpha_i \in \mathbb{C} \bigg\}.$$
\end{definition}
Note that $\U_n^Z$ is a (multiplicative) subgroup of $\U_n$.

\begin{definition} \label{def:BnZ}
An $n$-qubit unitary operator $G$ is a \emph{$Z$-bias-preserving gate} if
$$\forall P \in \mathcal{P}_n^Z \ \exists A \in \mathbb{U}_n^Z \text{ such that } GP = AG.$$
The set of all such gates is denoted by $\mathbb{B}_n^Z$.
\end{definition}

The above definition says that if an error of the form $P \in \mathcal{P}_n^Z$ occurs in an $n$-qubit system in a quantum circuit right before the application of the gate $G$, the equivalent error at the output of the gate is given by $U \in \mathbb{U}_n^Z$ acting on the system of $n$-qubits. These are errors that can be  detected using a syndrome circuit for detecting $Z$ errors. This definition of bias-preserving gates is also considered in an earlier work of \cite{guillaud2019repetition} for the case of $2$-qubit systems, where it is also proved that such a set of bias-preserving gates is a closed group.

\begin{remark} \label{rem:H}
The $X$-bias-preserving gate set $\B_n^X$, defined analogously (see Definition~3 in \cite{fellous2023scalable}), is related to $\B_n^Z$ via conjugation by Hadamards: $B_n^X = H^{\otimes n} \B_n^Z H^{\otimes n}$. Note also that $H^{\otimes n}$ is not in $\B_n^Z$ (or for that matter, in $\B_n^X$), as it does not satisfy the requirement in Definition~\ref{def:BnZ}: for example, $H^{\otimes n} Z^{\otimes n} = X^{\otimes n} H^{\otimes n}$.
\end{remark}

Fellous-Asiani et al.\ \cite{fellous2025scalable} gave a useful characterization of $X$-bias-preserving gates, showing that these are precisely those unitaries that act as permutations of the $X$-basis states while also allowing the introduction of local phases. This result immediately translates to a characterization of $Z$-bias-preserving gates by replacing $X$-basis states with $Z$-basis (i.e., computational basis) states.


\begin{theorem} \label{thm:bnz-permutation-fellous} A unitary $V$ is in $\mathbb{B}_n^Z$ if and only if for all $\s\in \{0,1\}^n$, there exists a real phase $\phi_{\s,V}$ such that $V\ket{\s} = e^{i\phi_{\s,V}}\ket{\sigma_V(\s)}$ where $\sigma_V$ is some permutation acting on $\{0,1\}^n$.
\end{theorem}

 In this work, we undertake a further study of the structure and properties of bias-preserving gates, building upon \cite{fellous2025scalable}.
The structure of the rest of the paper is as follows. In Section~\ref{sec:properties}, we give a characterization of bias-preserving gates based on their decomposition in terms of linear combinations of Pauli matrices. We further use this characterization to sketch an alternate proof of Theorem~\ref{thm:bnz-permutation-fellous}. In Section~\ref{sec:discrete-gate-set}, we show that a finite set of gates $\{X,R_z(\theta),CNOT,CCNOT\}$ can be used to approximate the set of all $Z$-bias-preserving gates to any desired accuracy. Next, in Section~\ref{sec:bipr-at-logical-level}, we draw the connection between bias-preserving gates and incoherent unitaries from coherence resource theory. We define the idea of equicoherent encoding maps and note that Calderbank-Shor-Steane (CSS) codes admit such encoding maps. Finally, we prove that for CSS codes, bias-preserving gates at the logical level can be realized by applying bias-preserving gates at the physical level.

\section{Characterizations of Bias-Preserving Gates}\label{sec:properties}

In this section, we first observe that the set of $Z$-bias-preserving gates is the normalizer in $\U_n$ of the subgroup of $Z$-type unitary operators. 
We then give an alternative characterization of bias-preserving gates in terms of linear combinations of Pauli gates, and use it to give an alternative proof of Theorem~\ref{thm:bnz-permutation-fellous}.

Recall that the normalizer, $\mathcal{N}(\mathcal{S})$, of a subgroup $\mathcal{S}$ of $\U_n$ is the set of all $G \in \U_n$ such that for each  $S \in \mathcal S$, there exists an $S' \in \mathcal{S}$ for which $GS = S'G$.
\begin{proposition} \label{prop:bnx-equal-to-normalizer-of-unx}
    The set $\mathbb{B}_n^Z$ is the normalizer in $\mathbb{U}_n$ of the subgroup $\mathbb{U}_n^Z$, i.e.,  $\mathbb{B}_n^Z = \mathcal{N}(\mathbb{U}_n^Z)$. Consequently, $B_n^Z$ is a closed subgroup of $\U_n$. 
\end{proposition}

\begin{proof}
The first statement is more or less definitional. If $G \in \mathcal{N}(\U_n^Z)$, then since any $P \in \mathcal{P}_n^Z$ is also in $\U_n^Z$, there exists an $A \in \U_n^Z$ such that $GP = AG$; hence, $G \in \B_n^Z$. Conversely, if $G \in \B_n^Z$, then for any $U = \sum_i \alpha_i P_i \in \U_n^Z$, there exist $A_i \in \U_n^Z$ such that $GU = \sum_i \alpha_i (GP_i) = \sum_i \alpha_i (A_i G) = AG$, where $A := \sum_i \alpha_i A_i$. Note that, since $A = GUG^\dag$, it is in $\U_n$, and moreover, since each $A_i$ is a linear combination of gates in $\mathcal{P}_n^Z$, so is $A$. Therefore, $A$ is in $\U_n^Z$, from which we conclude that $G \in \mathcal{N}(\U_n^Z)$. This proves that $\mathbb{B}_n^Z = \mathcal{N}(\mathbb{U}_n^Z)$.

The set $\unz$ is a closed subgroup of the unitary group $\U_n$, and the normalizer of a closed subgroup is closed. Hence, $B_n^Z$ is also a closed subgroup of $\U_n$. 
\end{proof}

As observed previously (Remark~\ref{rem:H}), $H^{\otimes n}$ is not in $\B_n^Z$. Consequently, $\B_n^Z$ is not dense in $\U_n$ for all $n$, which yields the corollary below to Proposition~\ref{prop:bnx-equal-to-normalizer-of-unx}. As defined in \cite{aharonov2003simple}, a set of quantum gates $\mathcal{S}$ is said to be \emph{strictly universal} if, for all sufficiently large $n$, the subgroup generated by $\mathcal{S}$ is (up to a global phase) dense in $\mathbb{U}_n$.

\begin{corollary}
\label{cor:BnZ_not_strictly_universal}
    The set of $Z$-bias-preserving gates is not strictly universal for quantum computation. 
\end{corollary}

In fact, bias-preserving gates are far from being a universal gate set. Not only is $H^{\otimes n}$ not in $\B_n^Z$, it cannot even be well-approximated by any gate in $B_n^Z$. Indeed, this can be quantified in terms of the worst-case error when an $n$-qubit gate $U$ is replaced by an $n$-qubit gate $V$ in an arbitrary quantum circuit (see \cite[Section~4.5.3]{nielsen_and_chuang}): $\mathcal{E}(U,V) := \max_{\ket\psi} \|(U-v) \ket\psi\|$, where $\| \cdot \|$ is the $L^2$-norm and the max is over $n$-qubit pure states. In order for $V$ to serve as a good approximation for $U$, we must have $\mathcal{E}(U,V)$ being arbitrarily small; but this is not possible when $U = H^{\otimes n}$ and $V$ is in $\B_n^Z$.

\begin{lemma} \label{lem:Hn}
For $U = H^{\otimes n}$ and any $V \in \B_n^Z$, we have $\mathcal{E}(U,V) \ge \sqrt{2(1-2^{-n/2})}$.
\end{lemma}
\begin{proof}
For $U = H^{\otimes n}$ and $V \in \B_n^Z$, we have 
\begin{align*} 
\mathcal{E}(U,V) \ &\ge \ \|U\ket{0}^{\otimes n} - V\ket{0}^{\otimes n}\| \\
&= \ \|\ket{+}^{\otimes n} - e^{i\phi} \ket{\mathrm{v}}\|,
\end{align*}
for some computational basis state $\ket{\mathrm{v}}$ and phase factor $e^{i\phi}$, by Theorem~\ref{thm:bnz-permutation-fellous}. 
Now, use the fact that for pure states $\ket\varphi$ and $\ket\psi$, we have $\|\ket\varphi - \ket\psi\|^2 = 2[1-\mathrm{Re}\braket{\varphi|\psi}] \ge 2 \, [\, 1-| \! \braket{\varphi|\psi} \! | \, ]$. Thus, 
$$
\|\ket{+}^{\otimes n} - e^{i\phi} \ket{\mathrm{v}}\|^2 \ \ge \ 2 \, [1 - 2^{-n/2}],
$$
from which the result follows. 
\end{proof}

We next introduce a novel way of viewing bias-preserving gates, in terms of their decompositions as linear combinations of Pauli matrices. Recall that for $\u = (u_1,u_2,\ldots,u_n) \in \{0,1\}^n$, we set $Z_\u := \bigotimes_i {Z_i}^{u_i}$; we define $X_\u := \bigotimes_i {X_i}^{u_i}$ analogously.
Since the $n$-qubit Pauli matrices form a basis of the space of all $2^n\times 2^n$ matrices, we can write an arbitrary $G \in \U_n$ uniquely as
$G = \sum_{\v\in \{0,1\}^n} \sum_{\u\in \{0,1\}^n} \alpha_{\u,\v} Z_\u X_\v$. Thus, setting $A_\v := \sum_{\u\in \{0,1\}^n}\alpha_{\u,\v}Z_\u X_\v$, we have $G = \sum_{\v\in \{0,1\}^n}A_\v$. We will refer to this as the $ZX$-decomposition of the unitary $G$.

In the following, we give an alternate characterization of bias-preserving gates in terms of their $ZX$-decompositions.

\begin{theorem} \label{thm:xz}
    An $n$-qubit gate $G$ with $ZX$-decomposition $G = \sum_{\v\in\{0,1\}^n}A_\v$ is $Z$-bias-preserving if and only if the following two conditions hold:
    \begin{enumerate}
        \item {[Orthogonality]} \ For all $\v,\w\in \{0,1\}^n$ such that $v\neq w$, $$A_\v A_\w^\dagger = 0.$$
        \item {[Completeness]} \ $\sum_{\v\in\{0,1\}^n}A_\v A_\v^{\dagger} = I.$
    \end{enumerate}
\end{theorem}

\begin{proof}
    For any gate $G \in \U_n$, we have, via its $ZX$-decomposition, 
    \begin{align}
    GZ_{\x}G^\dagger
    &= Z_{\x} \cdot \sum_{\v,\w\in \{0,1\}^n} (-1)^{\x\cdot\v}A_{\v}A_{\w}^\dagger \nonumber \\
    &= Z_{\x} \bigg(\sum_{\v\in \{0,1\}^n}(-1)^{\x\cdot\v}A_{\v}A_{\v}^\dagger 
    \nonumber \\
    & \ \ \ \ \ \ \ \ \ +
    \sum_{\substack{\v,\w \in \{0,1\}^n: \\ \v \ne \w}}(-1)^{\x\cdot\v}A_{\v} A_{\w}^\dagger\bigg). \label{eq:xz-lemma-line}
    \end{align}

    Now, note that 
    $$A_{\v}A_{\w}^{\dag} = \sum_{\u,\u'\in \{0,1\}^n}\alpha_{\u,\v} \alpha^*_{\u',\w} Z_\u X_{\v \oplus \w} Z_{\u'}.$$ Thus, $A_{\v}A_{\v}^{\dag}$ involves only $Z$ terms, while for $\v \ne \w$, $A_{\v}A_{\w}^{\dag}$ contains $X$ terms. Therefore, if the orthogonality condition $A_{\v}A_{\w}^{\dag} = 0$ holds for all $\v \ne \w$, then $GZ_{\x}G^{\dag}$ contains no $X$ terms, so that $G Z_{\x} G^{\dag} \in \U_n^Z$. Thus, $G$ is $Z$-bias-preserving. (The completeness condition simply enforces $GG^{\dag} = I$.)
    
    Conversely, suppose that $G$ is $Z$-bias-preserving. Then $G Z_{\x} G^{\dag}$ must be in $\U_n^Z$, so that there can be no $X$ terms in \eqref{eq:xz-lemma-line}. Therefore, re-writing the summation in \eqref{eq:xz-lemma-line} as 
    $$
    \sum_{\substack{\z \in \{0,1\}^n \\ \z \ne \mathbf{0}}} \sum_{\v \in \{0,1\}^n} (-1)^{\x\cdot\v} A_{\v}A_{\v \oplus \z}^{\dag},
    $$
    we infer that for each $\z \in \{0,1\}^n$, $\z \ne \mathbf{0}$, we must have
    \begin{equation}
    \sum_{\v \in \{0,1\}^n} (-1)^{\x\cdot\v} A_{\v}A_{\v \oplus \z}^{\dag} = 0,
    \label{eq:Xz_contrib}
    \end{equation}
    since the left-hand side constitutes precisely the $X_{\z}$ contribution to $GZ_{\x}G^{\dag}$. 
    
    Moreover, the equality in \eqref{eq:Xz_contrib} must hold for every $\x \in \{0,1\}^n$, since $G Z_{\x} G^{\dag} \in \U_n^Z$ for all $\x \in \{0,1\}^n$. Thus, for every $\z \in \{0,1\}^n$, $\z \ne \mathbf{0}$, we must have $2^n$ equations of the form \eqref{eq:Xz_contrib}, one for each $\x \in \{0,1\}^n$. These $2^n$ equations can be expressed as a single matrix equation: $\textbf{H} \cdot \textsf{V}_{\z} = 0$,
    where $\textbf{H} = {[(-1)^{\x\cdot\v}]}_{\x,\v \in \{0,1\}^n}$ is a $2^n \times 2^n$ matrix, and $\textsf{V}_{\z}$ is a $2^n \times 1$ column vector with entries $A_{\v}A_{\v \oplus \z}^{\dag}$, $\v \in \{0,1\}^n$. In fact, $\textbf{H}$ is the Hadamard matrix 
    $\begin{bmatrix}
    1 & 1 \\
    1 & -1
    \end{bmatrix}^{\otimes n},$
    which is invertible. Therefore, $\textbf{H} \cdot \textsf{V}_{\z} = 0$ implies that $\textsf{V}_{\z} = 0$, and this holds for all $\z \ne \mathbf{0}$. From this, we conclude that the orthogonality condition holds: $A_{\v}A_{\w}^{\dag} = 0$ for all $\v \ne \w$.

    Returning to \eqref{eq:xz-lemma-line}, we now have 
    $$
    GZ_{\x}G^\dagger = Z_{\x} \cdot \sum_{\v\in \{0,1\}^n}(-1)^{\x\cdot\v}A_{\v}A_{\v}^\dagger.
    $$
    Setting $\x = \mathbf{0}$ above, we obtain 
    $$
    \sum_{\v\in \{0,1\}^n} A_{\v}A_{\v}^\dagger = GG^{\dag} = I.
    $$
    which is the completeness condition.
\end{proof}

We next demonstrate how the characterization of $Z$-bias-preserving gates stated in Theorem~\ref{thm:bnz-permutation-fellous} can be recovered from the above theorem.

\begin{theorem}[Re-statement of Theorem~\ref{thm:bnz-permutation-fellous}] \label{thm:partition}
An $n$-qubit gate $G$ is $Z$-bias-preserving iff it is of the form $G = \sum_{\s \in \{0,1\}^n} e^{i\phi_{\s}} \ket{\sigma(\s)} \bra{\s}$ for some choice of real phases $\phi_{\s}$ and some permutation $\sigma$ acting on $\{0,1\}^n$. 
\end{theorem}

\begin{proof}
For the ``if'' part, a direct verification of the $Z$-bias-preserving property of $G = \sum_{\s} e^{i\phi_{\s}} \ket{\sigma(\s)} \bra{\s}$ is possible, as was done in \cite{fellous2023scalable} for their proof of Theorem~\ref{thm:bnz-permutation-fellous}. We prove the ``only if'' part here. For this, we utilize the spectral decomposition of $Z_{\u}$, which is of the form:
$$
Z_{\u} = \sum_{\s \in \{0,1\}^n} a_{\s,\u} \ket{\s}\bra{\s}, \ \text{ with } a_{\s,\u} = \pm 1 \  \forall \s.
$$ 
This allows us to write $A_\v = \sum_{\u}\alpha_{\u,\v}Z_\u X_\v$ as
\begin{eqnarray}
A_{\v} \ & = & \ \left(\sum_{\s \in \{0,1\}^n} \beta_{\s,\v} \ket{\s}\bra{\s}\right) X_{\v} \notag \\
& = & \sum_{\s \in S_{\v}} \beta_{\s,\v} \ket{\s} \bra{\sigma_{\v}(\s)} 
\label{eq:Av}
\end{eqnarray}
where $\ket{\sigma_{\v}(\s)} = X_{\v}\ket{\s} = \ket{\v \oplus \s}$,  $\beta_{\s,\v} = \sum_{\u \in \{0,1\}^n} a_{\s,\u} \alpha_{\u,\v}$, and $S_{\v} = \{\s: \beta_{\s,\v} \ne 0\}$. In particular, note that $\sigma_{\v}(\s) = \v \oplus \s$ so that $|\sigma_{\v}(S_{\v})| = |S_{\v}|$. 

It is easy to verify that, for $\v \ne \w$, $A_{\v}A_{\w}^{\dag} = 0$ iff $\sigma_{\v}(S_{\v}) \cap \sigma_{\w}(S_{\w}) = \emptyset$. Moreover, for the completeness condition ($\sum_{\v}A_{\v}A_{\v}^\dag = I$) to hold, it must be the case that $\bigcup_{\v} S_{\v} = \{0,1\}^n$. Thus, if $G = \sum_{\v}A_{\v}$ is $Z$-bias-preserving, then by Theorem~\ref{thm:xz}, the orthogonality and completeness properties hold, so that we have
\begin{align*}
2^n = \left|\bigcup_{\v} S_{\v}\right| & \le \sum_{\v} |S_{\v}| \\ &= \sum_{\v} |\sigma_{\v}(S_{\v})| = \left|\bigcup_{\v} \sigma_{\v}(S_{\v})\right| \le 2^n.
\end{align*}
Therefore, all inequalities in the sequence above are equalities, from which we infer that both $\{S_{\v}: \v \in \{0,1\}^n\}$ and $\{\sigma_{\v}(S_{\v}): \v \in \{0,1\}^n\}$ are partitions of $\{0,1\}^n$. It then follows from \eqref{eq:Av} that $G = \sum_{\v} A_{\v}$ is of the form $\sum_{\s \in \{0,1\}^n} \beta_{\s} \ket{\s} \bra{\tilde{\sigma}(\s)}$ for some $\beta_{\s} \in \C$ and some permutation $\tilde{\sigma}$ acting on $\{0,1\}^n$. Equivalently, via a change of variable, $G = \sum_{\s \in \{0,1\}^n} \gamma_{\s} \ket{\sigma(\s)}\bra{\s}$, where $\sigma = {\tilde{\sigma}}^{-1}$. Finally, for $GG^\dag = I$ to hold, we must have $|\gamma_{\s}|^2 = 1$ for all $\s$, i.e., $\gamma_{\s}$ must have the form $e^{i\phi_{\s}}$. 
\end{proof}

Thus, the characterization of bias-preserving gates by Fellous-Asiani et al.\ \cite{fellous2023scalable} can be recovered using our $ZX$-decomposition framework.

\section{A Discrete Generating Set for Bias-Preserving Gates}\label{sec:discrete-gate-set}

Since the set of bias-preserving gates forms an uncountably infinite subgroup of the unitary group, it would be useful to approximate arbitrary $n$-qubit bias-preserving gates using a discrete set of gates that is also bias-preserving. In this section we give a procedure to do so. The technique mostly relies on procedures that approximate arbitrary unitary matrices using a discrete gate set \cite[Section~4.5]{nielsen_and_chuang}. We also note that this set was also mentioned in \cite[Section~5]{shepherd2008instantaneous} in the context of $Z$-networks.

\begin{lemma}\label{lemma:b1z-from-discrete}
Any gate $G\in \bOz$ can be approximated by gates from the set $\{R_z(\theta), X\}$, where $R_z(\theta)$ is a rotation about the $z$-axis (in the Bloch sphere) by the angle $\theta$, which can be taken to be any irrational multiple of $2\pi$.
\end{lemma}
\begin{proof}
    Since $G\in \bOz$, we can write it as either of the following matrices (written in the $Z$ basis),
    $$V_1 := 
    \begin{pmatrix}
        0 & e^{i\phi_1} \\
        e^{i\phi_2} & 0 \\
    \end{pmatrix},
    V_2 := 
    \begin{pmatrix}
        e^{i\phi_3} & 0 \\
        0 & e^{i\phi_4} \\
    \end{pmatrix}.
    $$
    It is enough to prove the lemma for $V_2$, as $V_1$ can be obtained using the $X$ gate, that is, $V_1 = V_2X$.
    Now note that applying $V_2$ is equivalent (up to a global phase) to applying the rotation gate,
    $$
     P_Z(\phi)
     := \begin{pmatrix}
        1 & 0 \\
        0 & e^{i\phi}
    \end{pmatrix}. 
     $$
     where $\phi := \phi_4 - \phi_3$ in $V_2$. 
     Therefore, it is sufficient to use $X, P_Z(\phi)$ for arbitrary $\phi\in \mathbb{R}$.
     
     
     To complete the proof, we note that $P_Z(\phi)$, for any $\phi$ can be approximated by a single $P_Z(\theta)$ gate, where $\theta$ can be taken to be any irrational multiple of $2\pi$.
\end{proof}

\begin{theorem}
    Any $n$ qubit $Z$-bias-preserving gate can be approximated using bias-preserving gates from the gate set
    $$\{X, R_z(\theta), CNOT, CCNOT\},$$
    where $\theta$ is any irrational multiple of $2\pi$ and $CCNOT$ is the Toffoli gate.
\end{theorem}

\begin{proof}
First note that all gates in the proposed gate set are $Z$-bias-preserving, as they are permutation matrices with phases, in the $Z$ basis. Let $G\in \bnz$ be the gate we want to approximate. Since $G$ is a permutation with phases, we can write it as $G = DP$, where $D = \text{diag}(e^{i\phi_1}, \ldots, e^{i\phi_{2^n}})$ and $P$ is a permutation matrix over the computational basis. Since $D$ is a non-singular diagonal matrix and $P$ is a permutation matrix, they are both unitary. We will construct these separately in the rest of the proof.

Note that we can write $P = \sigma_1 \cdots \sigma_t$ for some $t\in\mathbb{N}$, where for all $i\in [t]$, $\sigma_i$ is a transposition (i.e. a 2-unitary), which can be implemented exactly using $\{X,CNOT, CCNOT\}$ using standard techniques from \cite{nielsen_and_chuang}. Similarly, $D$ can be approximated using gates in the proposed gate set. Since $D$ is non-singular, we can write it as the product $D = D_1 \cdots D_{2^n}$ where $D_j$ is the diagonal matrix with element $e^{i\theta_j}$ at the $j$-th index and 1 on the other diagonal indices. $D_j$ is a 1-level unitary which can be approximated using Lemma \ref{lemma:b1z-from-discrete} and standard techniques from \cite{nielsen_and_chuang}.
\end{proof}

\section{Equivalence between Logical and Physical Bias-preserving gates}\label{sec:bipr-at-logical-level}

In this section, we first note the equivalence of bias-preserving gates and incoherent unitaries from coherence resource theories \cite{chitambar2016critical, streltsov2017colloquium}. We then introduce the notion of equicoherent basis/encodings and note that well known class of CSS codes are equicoherent. We prove that for CSS codes under their standard encoding map, the set of bias-preserving gates acting on the physical space of $n$-qubits are equivalent to the bias-preserving gates acting on the logical space of $k$-qubits and vice-versa. We also show that for a certain class of non-equicoherent encoding maps, the set of bias-preserving gates at the physical and the logical level are not equivalent.

 To note that bias-preserving gates and incoherent unitaries are identical, we recall the definitions of incoherent unitaries and coherence rank.
 
\begin{definition}[Incoherent Unitaries \cite{jones2024hadamard, streltsov2017colloquium}]
A unitary $U$ is said to be an $n$ qubit \emph{incoherent unitary} if it can be written as
\begin{equation*}
U = \sum_{j =1}^d e^{i\theta_j}\ket{\pi(j)}\bra{j},
\end{equation*}
where $\pi$ is a permutation on the reference basis $\{\ket{j}\}_{j=1}^d$ of the system.
\end{definition}

It directly follows from Theorem~\ref{thm:bnz-permutation-fellous} that $G$ is a $Z$-bias-preserving gate if and only if it is an incoherent unitary in the $Z$ basis.

\begin{definition}[Coherence Rank $\chi$ \cite{streltsov2017colloquium}]
The \emph{coherence rank} $\chi:\mathcal{H}\rightarrow \mathbb{N}$ of a pure state in Hilbert space $\mathcal{H}$ in a fixed basis $\{\ket{c_i}\}_{i=1}^d$ is the number of non-zero terms needed in the expansion of that state in that basis. That is,
$$\chi(\ket{\psi}) := 
|\{i: a_i \neq 0\}|, \text{where} \ket{\psi} = \sum_i a_i \ket{c_i}.
$$

We take the coherence rank of the 0 vector to be 0. The Hilbert space under consideration while using the notation of $\chi$ will be clear from the context.
\end{definition}
The following proposition establishes that incoherent unitaries/bias-preserving gates preserve the coherence rank of pure states.

\begin{proposition}[\cite{jones2024hadamard}] \label{prop:coherence_rank}
A unitary $G$ is $Z$-bias-preserving, i.e., $G \in \mathbb{B}_n^Z$ if and only if
$$ \chi(G\ket{x}) = \chi(\ket{x}), \forall x \in \{0,1\}^n.$$
 \end{proposition}

Now, we will set up the notation of quantum error-correcting codes and also introduce the notion of equicoherent basis/encoding.   
 An $[[n,k,d]]$ quantum error correcting code $\mathcal{Q}$ is associated with an encoding map $\phi_{\text{enc}}:\mathcal{H}_2^{\otimes k}\rightarrow \mathcal{H}_2^{\otimes n}$, with its inverse defined as $\phi_{\text{enc}}^{-1}:\mathcal{Q} \rightarrow \mathcal{H}_2^{\otimes k}$. A state $\ket{\psi}$ is a quantum state in Hilbert space $\mathcal{H}$, which is over $n$ or $k$ qubits depending on the context. In the context of codewords, $\ket{\psi}_L$ is an $n$ qubit quantum state which represents the $k$ qubit state $\ket{\psi}$ at a logical level. So we have, $\ket{\psi}_L = \phi_{\text{enc}}\ket{\psi} = \sum_{x\in \{0,1\}^n} \alpha_x\ket{x}$, where $\sum_x |\alpha_x|^2 = 1$. The set of logical operators $\mathcal{L}_{\mathcal{Q}}$ is the set of gates acting on the physical space of $n$ qubits, which map the codespace $\mathcal{Q}$ to itself, i.e.,
 \begin{equation*}
\mathcal{L}_{\mathcal{Q}} = \{ G \in \mathbb{U}_n \mid \ G \ket{\psi} \in \mathcal{Q}, \forall \ket{\psi} \in \mathcal{Q}\}.
\end{equation*}

We are interested in the problem of characterizing the class of quantum error-correcting codes for which applying bias-preserving gates at the physical level is equivalent to action of corresponding bias-preserving gates at the logical level. Based on Proposition \ref{prop:coherence_rank}, we can expect that the encoding map has to preserve the coherence rank between states in the physical space, which are images of computational basis states under the map $\phi_{\text{enc}}$.  
 Motivated by this, we introduce the notion of \emph{equicoherent encoding maps}. For convenience, we fix the computational basis, but these can be extended to arbitrary bases easily. We define $\mathcal{T}$ to be the function that takes as input a pure state and returns a set of all terms with non-zero coefficients when expanded in the computational basis. For example, consider a pure state in a $3$ qubit system given by $\ket{\psi} = \frac{1}{2} \left(\ket{000} + \ket{011} + \ket{101} + \ket{110} \right)$, then $\mathcal{T}(\ket{\psi}) = \{ 000, 011, 101, 110\}$.

\begin{definition} \label{def:equicoherent} 
Let $\mathcal{Q}$ be a quantum error correcting code with encoding map $\phi_{\text{enc}}$. The representation of the basis states of the $k$-qubit logical space in the $n$-qubit physical codespace is said to form an \emph{equicoherent basis} under the encoding map if:
\begin{enumerate}
    \item For $x,y \in \{0,1\}^k$ where $|x \big>, |y \big>$ are two arbitrary basis states in the logical space, we have, $$\chi(\ket{x}_L) = \chi(\ket{y}_L).$$
    \item For $x, y\in \{0,1\}^k$ where $|x \big>, |y \big>$ are two arbitrary basis states in the logical space, we have, $$\mathcal{T}(\ket{x}_L)\ \cap\ \mathcal{T}(\ket{y}_L) = \emptyset,$$
  where for $z \in \{0,1\}^k$, we have $\ket{z}_L := \phi_{\text{enc}}\ket{z}$.  
\end{enumerate}
The encoding map $\phi_{\text{enc}}$ is then referred to as an \emph{equicoherent encoding}.
\end{definition}

We highlight some important properties of equicoherent encodings that we shall use later in the section.
\begin{proposition} \label{prop:prop-of-equicoherent-enc}
    Let $\mathcal{Q}$ be a quantum error correcting code with equicoherent encoding $\phi_{\text{enc}}$. Let $l$ be the coherence rank of the computational basis states, i.e., $\chi(\ket{x}_L) = l \ \forall x \in \{0,1\}^k$. Then, the following hold:
    \begin{enumerate}
        \item Let $\ket{\psi} = \sum_{x\in \{0,1\}^k}\alpha_x \ket{x}$ be a $k$-qubit state. Then, $\chi(\ket{\psi}_L) = \chi(\sum_{x\in \{0,1\}^k}\alpha_x\ket{x}_L) = \sum_{x\in \{0,1\}^k}\chi(\ket{x}_L)$. 
        \item If $\ket{\psi} \in \mathcal{Q}$ is such that $\chi(\ket{\psi}) = l$, then $\phi_{\text{enc}}^{-1} \ket{\psi} = e^{i\beta} \ket{x}$ for some $x\in \{0,1\}^k$ and some phase $\beta$.
        \item Let $\ket{\psi}, \ket{\varphi}$ be $k$-qubit states. Then, $\chi(\ket{\psi}) = \chi(\ket{\varphi})$ if and only if $\chi(\phi_{\text{enc}}\ket{\psi}) = \chi(\phi_{\text{enc}}\ket{\varphi})$.
    \end{enumerate}
\end{proposition}
The proofs of these properties are straightforward. \\

The class of CSS codes is an important example of quantum codes having an equicoherent basis. A CSS code \cite{calderbank1996good,steane1996error} is a quantum stabilizer code \cite{gottesman1997stabilizer} which can be defined using two classical linear codes $\mathcal{C}_1 \subset \mathcal{C}_2$. The $X$-type stabilizer generators are obtained from the rows of the parity-check matrix of $\mathcal{C}_1^\perp$, and the $Z$-type stabilizer generators are obtained from rows of the parity-check matrix of $\mathcal{C}_2$. The inclusion $\mathcal{C}_1 \subset \mathcal{C}_2$ ensures that the $X$-type stabilizers and $Z$-type stabilizers commute. The CSS code $\mathcal{Q}$ is the common +1-eigenspace of all the stabilizers. The number of logical qubits it encodes is $k = \dim(\mathcal{C}_2/\mathcal{C}_1) = \dim(\mathcal{C}_2) - \dim(\mathcal{C}_1)$.

It is a standard fact (see e.g., \cite[Section~10.4.2]{nielsen_and_chuang}) that the quantum states $\ket{x+\mathcal{C}_1} := \frac{1}{\sqrt{|\mathcal{C}_1|}}\sum_{y\in \mathcal{C}_1} \ket{x+y}$, $x \in \mathcal{C}_2/\mathcal{C}_1$, constitute a basis of $\mathcal{Q}$. This basis is clearly equicoherent. A \emph{standard encoding} of the CSS code $\mathcal{Q}$ is any equicoherent encoding that maps the computational basis states of the logical space to this equicoherent basis.

Consider the set of bias-preserving gates applied at the logical level $\mathbb{B}_k^Z$. The transformations at the physical level that are equivalent to bias-preserving gates at the logical level form a subset of the logical operators $\mathcal{L}_{\mathcal{Q}}$. We are interested in the question of when the logical operators are bias-preserving gates at the physical level. Hence, we consider the following set of transformations:
\begin{equation*}
\mathbb{B}_{n,\mathcal{Q}}^{Z} = \{ G|_{\mathcal{Q}} \mid G \in \mathbb{B}_n^Z \cap \mathcal{L}_{\mathcal{Q}} \}.
\end{equation*}
The notation $G|_{\mathcal{Q}}$ refers to the restriction of the map $G$ to the subspace $\mathcal{Q}$. We will refer to $\mathbb{B}_{n,\mathcal{Q}}^{Z}$ as the $\mathcal{Q}$-restricted set of bias-preserving gates at the physical level. The elements in 
$\mathbb{B}_{n,\mathcal{Q}}^{Z}$ (since they are also restricted logical operators) can be expressed as $\phi_{\text{enc}} \circ G \circ \phi_{\text{enc}}^{-1}$ for some $G \in \mathbb{U}_k$.

\begin{theorem}
    For a CSS code $\mathcal{Q}$ with a standard encoding $\phi_{\text{enc}}$, we have 
    $$
    \mathbb{B}_{n,\mathcal{Q}}^Z = \phi_{\text{enc}} \mathbb{B}_k^Z \phi_{\text{enc}}^{-1} := \{\phi_{\text{enc}} \circ G \circ \phi_{\text{enc}}^{-1} : G \in \B_k^Z\}.
    $$
    Thus, the set of bias-preserving gates at the logical level is isomorphic to the $\mathcal{Q}$-restricted set of bias-preserving gates at the physical level. 
\end{theorem}

\begin{proof}
Consider a $\mathcal{Q}$-restricted gate $\bar{G} \in \mathbb{B}_{n,\mathcal{Q}}^Z$. $\bar{G}$ maps $\mathcal{Q}$ onto $\mathcal{Q}$, and also preserves the coherence rank. By Property~2 of Proposition~\ref{prop:prop-of-equicoherent-enc}, it follows that, for all  $s \in \{0,1\}^k$, $\bar{G}$ maps $\phi_{\text{enc}}\ket{s}$ to $\phi_{\text{enc}}\left(e^{i\beta_s}\ket{\pi(s)}\right)$, where $\beta_s$ is a phase and $\pi$ is some permutation of $\{0,1\}^k$. In other words, $G := \phi_{\text{enc}}^{-1} \circ \bar{G} \circ \phi_{\text{enc}}$ maps each $k$-qubit computational basis state $\ket{s}$ to $e^{i\beta_s}\ket{\pi(s)}$. By Theorem~\ref{thm:bnz-permutation-fellous}, $G$ is in $B_k^Z$. We have thus found a $G \in \B_k^Z$ such that $\bar{G} = \phi_{\text{enc}} \circ G \circ \phi_{\text{enc}}^{-1}$.
    
 For the converse, we would like to show that for each logical bias-preserving gate $G \in \mathbb{B}_{k}^{Z}$, there exists a unique $\mathcal{Q}$-restricted bias-preserving gate at the physical level, which is equivalent. Consider $G \in \mathbb{B}_{k}^{Z}$. We have that
$$ G\ket{s} = e^{i \beta_s} \ket{\pi(s)}, $$ where $\pi$ is some permutation of $\{0,1\}^k$. We will now identify the bias-preserving gate in $\bnz$, which will result in identical map in the codespace $\mathcal{Q}$. Let $\gamma_s$ denote a partial permutation which maps elements in $\mathcal{T}(\ket{s}_L)$ to those in $\mathcal{T}(\ket{\pi(s)}_L)$. Since CSS codes satisfy the following condition:
$ \mathcal{T}(\ket{s}_L) \cap \mathcal{T}(\ket{t}_L) = \phi, \forall s \neq t$, we can combine all the partial permutations $\gamma_s$ to obtain a permutation $\gamma_G$ which acts on $n$ qubit computational basis states. The multiplicative phase factor associated with each partial permutation is $e^{i\beta_s}$. The permutation $\gamma_G$ and the phase factors $\{e^{i\beta_s}, s \in \{0,1\}^k \}$, together define a bias-preserving gate in $\bnz$. Let the gate be denoted by $\hat{G}$. Note that $\hat{G}$ satisfies 
\begin{equation*}
\phi_{\text{enc}}G = \hat{G} \phi_{\text{enc}}.
\end{equation*}

Also, note that $\hat{G}$ is not unique in general but $\hat{G}|_{\mathcal{Q}}$ is unique because $\phi_{\text{enc}}G = \hat{G} \phi_{\text{enc}}$. Hence, $\hat{G}|_{\mathcal{Q}}$ is given by
$\hat{G}|_{\mathcal{Q}} = \phi_{\text{enc}} \circ G \circ \phi_{\text{enc}}^{-1}$.
\end{proof}

We end this section by noting that the above theorem fails to hold for codes with non-equicoherent encodings. Indeed, suppose that $\mathcal{Q}$ is a quantum code with an encoding map that does not satisfy Condition~1 in Definition~\ref{def:equicoherent}. Then, there exist computational basis states $\ket{s}$ and $\ket{t}$ in the $k$-qubit Hilbert space such that $\chi({\ket{s}_L}) \neq \chi({\ket{t}_L})$. Consider a bias-preserving gate $G \in \mathbb{B}_k^Z$ that maps $\ket{s}$ to $\ket{t}$. Since $\chi({\ket{s}_L}) \neq \chi({\ket{t}_L})$, there is no corresponding bias-preserving gate in $\bnz$ that maps $\ket{s}_L$ to $\ket{t}_L$. Thus, there is no $\mathcal{Q}$-restricted bias-preserving gate acting on the physical qubits that realizes the action of $G$ on the logical qubits.

\section*{Acknowledgement} This work was supported in part by the grant DST/INT/RUS/RSF/P-41/2021 from the Department of Science and Technology, Government of India. Aryaman Manish Kolhe would like to acknowledge helpful discussions with Mahathi Vempati.

\bibliography{references}{}

\begin{thebibliography}{26}%
\makeatletter
\providecommand \@ifxundefined [1]{%
 \@ifx{#1\undefined}
}%
\providecommand \@ifnum [1]{%
 \ifnum #1\expandafter \@firstoftwo
 \else \expandafter \@secondoftwo
 \fi
}%
\providecommand \@ifx [1]{%
 \ifx #1\expandafter \@firstoftwo
 \else \expandafter \@secondoftwo
 \fi
}%
\providecommand \natexlab [1]{#1}%
\providecommand \enquote  [1]{``#1''}%
\providecommand \bibnamefont  [1]{#1}%
\providecommand \bibfnamefont [1]{#1}%
\providecommand \citenamefont [1]{#1}%
\providecommand \href@noop [0]{\@secondoftwo}%
\providecommand \href [0]{\begingroup \@sanitize@url \@href}%
\providecommand \@href[1]{\@@startlink{#1}\@@href}%
\providecommand \@@href[1]{\endgroup#1\@@endlink}%
\providecommand \@sanitize@url [0]{\catcode `\\12\catcode `\$12\catcode
  `\&12\catcode `\#12\catcode `\^12\catcode `\_12\catcode `\%12\relax}%
\providecommand \@@startlink[1]{}%
\providecommand \@@endlink[0]{}%
\providecommand \url  [0]{\begingroup\@sanitize@url \@url }%
\providecommand \@url [1]{\endgroup\@href {#1}{\urlprefix }}%
\providecommand \urlprefix  [0]{URL }%
\providecommand \Eprint [0]{\href }%
\providecommand \doibase [0]{https://doi.org/}%
\providecommand \selectlanguage [0]{\@gobble}%
\providecommand \bibinfo  [0]{\@secondoftwo}%
\providecommand \bibfield  [0]{\@secondoftwo}%
\providecommand \translation [1]{[#1]}%
\providecommand \BibitemOpen [0]{}%
\providecommand \bibitemStop [0]{}%
\providecommand \bibitemNoStop [0]{.\EOS\space}%
\providecommand \EOS [0]{\spacefactor3000\relax}%
\providecommand \BibitemShut  [1]{\csname bibitem#1\endcsname}%
\let\auto@bib@innerbib\@empty
\bibitem [{\citenamefont {Webster}\ \emph {et~al.}(2015)\citenamefont
  {Webster}, \citenamefont {Bartlett},\ and\ \citenamefont
  {Poulin}}]{webster2015reducing}%
  \BibitemOpen
  \bibfield  {author} {\bibinfo {author} {\bibfnamefont {P.}~\bibnamefont
  {Webster}}, \bibinfo {author} {\bibfnamefont {S.~D.}\ \bibnamefont
  {Bartlett}},\ and\ \bibinfo {author} {\bibfnamefont {D.}~\bibnamefont
  {Poulin}},\ }\bibfield  {title} {\bibinfo {title} {Reducing the overhead for
  quantum computation when noise is biased},\ }\href@noop {} {\bibfield
  {journal} {\bibinfo  {journal} {Physical Review A}\ }\textbf {\bibinfo
  {volume} {92}},\ \bibinfo {pages} {062309} (\bibinfo {year}
  {2015})}\BibitemShut {NoStop}%
\bibitem [{\citenamefont {Guti{\'e}rrez}\ \emph {et~al.}(2013)\citenamefont
  {Guti{\'e}rrez}, \citenamefont {Svec}, \citenamefont {Vargo},\ and\
  \citenamefont {Brown}}]{gutierrez2013approximation}%
  \BibitemOpen
  \bibfield  {author} {\bibinfo {author} {\bibfnamefont {M.}~\bibnamefont
  {Guti{\'e}rrez}}, \bibinfo {author} {\bibfnamefont {L.}~\bibnamefont {Svec}},
  \bibinfo {author} {\bibfnamefont {A.}~\bibnamefont {Vargo}},\ and\ \bibinfo
  {author} {\bibfnamefont {K.~R.}\ \bibnamefont {Brown}},\ }\bibfield  {title}
  {\bibinfo {title} {Approximation of realistic errors by {C}lifford channels
  and {P}auli measurements},\ }\href@noop {} {\bibfield  {journal} {\bibinfo
  {journal} {Physical Review A --- Atomic, Molecular, and Optical Physics}\
  }\textbf {\bibinfo {volume} {87}},\ \bibinfo {pages} {030302} (\bibinfo
  {year} {2013})}\BibitemShut {NoStop}%
\bibitem [{\citenamefont {Geller}\ and\ \citenamefont
  {Zhou}(2013)}]{geller2013efficient}%
  \BibitemOpen
  \bibfield  {author} {\bibinfo {author} {\bibfnamefont {M.~R.}\ \bibnamefont
  {Geller}}\ and\ \bibinfo {author} {\bibfnamefont {Z.}~\bibnamefont {Zhou}},\
  }\bibfield  {title} {\bibinfo {title} {Efficient error models for
  fault-tolerant architectures and the {P}auli twirling approximation},\
  }\href@noop {} {\bibfield  {journal} {\bibinfo  {journal} {Physical Review A
  --- Atomic, Molecular, and Optical Physics}\ }\textbf {\bibinfo {volume}
  {88}},\ \bibinfo {pages} {012314} (\bibinfo {year} {2013})}\BibitemShut
  {NoStop}%
\bibitem [{\citenamefont {Cai}\ and\ \citenamefont
  {Benjamin}(2019)}]{cai2019constructing}%
  \BibitemOpen
  \bibfield  {author} {\bibinfo {author} {\bibfnamefont {Z.}~\bibnamefont
  {Cai}}\ and\ \bibinfo {author} {\bibfnamefont {S.~C.}\ \bibnamefont
  {Benjamin}},\ }\bibfield  {title} {\bibinfo {title} {Constructing smaller
  {P}auli twirling sets for arbitrary error channels},\ }\href@noop {}
  {\bibfield  {journal} {\bibinfo  {journal} {Scientific reports}\ }\textbf
  {\bibinfo {volume} {9}},\ \bibinfo {pages} {11281} (\bibinfo {year}
  {2019})}\BibitemShut {NoStop}%
\bibitem [{\citenamefont {Aliferis}\ \emph {et~al.}(2009)\citenamefont
  {Aliferis}, \citenamefont {Brito}, \citenamefont {DiVincenzo}, \citenamefont
  {Preskill}, \citenamefont {Steffen},\ and\ \citenamefont
  {Terhal}}]{aliferis2009fault}%
  \BibitemOpen
  \bibfield  {author} {\bibinfo {author} {\bibfnamefont {P.}~\bibnamefont
  {Aliferis}}, \bibinfo {author} {\bibfnamefont {F.}~\bibnamefont {Brito}},
  \bibinfo {author} {\bibfnamefont {D.~P.}\ \bibnamefont {DiVincenzo}},
  \bibinfo {author} {\bibfnamefont {J.}~\bibnamefont {Preskill}}, \bibinfo
  {author} {\bibfnamefont {M.}~\bibnamefont {Steffen}},\ and\ \bibinfo {author}
  {\bibfnamefont {B.~M.}\ \bibnamefont {Terhal}},\ }\bibfield  {title}
  {\bibinfo {title} {Fault-tolerant computing with biased-noise superconducting
  qubits: a case study},\ }\href@noop {} {\bibfield  {journal} {\bibinfo
  {journal} {New Journal of Physics}\ }\textbf {\bibinfo {volume} {11}},\
  \bibinfo {pages} {013061} (\bibinfo {year} {2009})}\BibitemShut {NoStop}%
\bibitem [{\citenamefont {Cong}\ \emph {et~al.}(2022)\citenamefont {Cong},
  \citenamefont {Levine}, \citenamefont {Keesling}, \citenamefont {Bluvstein},
  \citenamefont {Wang},\ and\ \citenamefont {Lukin}}]{cong2022hardware}%
  \BibitemOpen
  \bibfield  {author} {\bibinfo {author} {\bibfnamefont {I.}~\bibnamefont
  {Cong}}, \bibinfo {author} {\bibfnamefont {H.}~\bibnamefont {Levine}},
  \bibinfo {author} {\bibfnamefont {A.}~\bibnamefont {Keesling}}, \bibinfo
  {author} {\bibfnamefont {D.}~\bibnamefont {Bluvstein}}, \bibinfo {author}
  {\bibfnamefont {S.-T.}\ \bibnamefont {Wang}},\ and\ \bibinfo {author}
  {\bibfnamefont {M.~D.}\ \bibnamefont {Lukin}},\ }\bibfield  {title} {\bibinfo
  {title} {Hardware-efficient, fault-tolerant quantum computation with rydberg
  atoms},\ }\href@noop {} {\bibfield  {journal} {\bibinfo  {journal} {Physical
  Review X}\ }\textbf {\bibinfo {volume} {12}},\ \bibinfo {pages} {021049}
  (\bibinfo {year} {2022})}\BibitemShut {NoStop}%
\bibitem [{\citenamefont {Puri}\ \emph {et~al.}(2020)\citenamefont {Puri},
  \citenamefont {St-Jean}, \citenamefont {Gross}, \citenamefont {Grimm},
  \citenamefont {Frattini}, \citenamefont {Iyer}, \citenamefont {Krishna},
  \citenamefont {Touzard}, \citenamefont {Jiang}, \citenamefont {Blais} \emph
  {et~al.}}]{puri2020bias}%
  \BibitemOpen
  \bibfield  {author} {\bibinfo {author} {\bibfnamefont {S.}~\bibnamefont
  {Puri}}, \bibinfo {author} {\bibfnamefont {L.}~\bibnamefont {St-Jean}},
  \bibinfo {author} {\bibfnamefont {J.~A.}\ \bibnamefont {Gross}}, \bibinfo
  {author} {\bibfnamefont {A.}~\bibnamefont {Grimm}}, \bibinfo {author}
  {\bibfnamefont {N.~E.}\ \bibnamefont {Frattini}}, \bibinfo {author}
  {\bibfnamefont {P.~S.}\ \bibnamefont {Iyer}}, \bibinfo {author}
  {\bibfnamefont {A.}~\bibnamefont {Krishna}}, \bibinfo {author} {\bibfnamefont
  {S.}~\bibnamefont {Touzard}}, \bibinfo {author} {\bibfnamefont
  {L.}~\bibnamefont {Jiang}}, \bibinfo {author} {\bibfnamefont
  {A.}~\bibnamefont {Blais}}, \emph {et~al.},\ }\bibfield  {title} {\bibinfo
  {title} {Bias-preserving gates with stabilized cat qubits},\ }\href@noop {}
  {\bibfield  {journal} {\bibinfo  {journal} {Science advances}\ }\textbf
  {\bibinfo {volume} {6}},\ \bibinfo {pages} {eaay5901} (\bibinfo {year}
  {2020})}\BibitemShut {NoStop}%
\bibitem [{\citenamefont {Aliferis}\ and\ \citenamefont
  {Preskill}(2008)}]{aliferis2008fault}%
  \BibitemOpen
  \bibfield  {author} {\bibinfo {author} {\bibfnamefont {P.}~\bibnamefont
  {Aliferis}}\ and\ \bibinfo {author} {\bibfnamefont {J.}~\bibnamefont
  {Preskill}},\ }\bibfield  {title} {\bibinfo {title} {Fault-tolerant quantum
  computation against biased noise},\ }\href@noop {} {\bibfield  {journal}
  {\bibinfo  {journal} {Physical Review A—Atomic, Molecular, and Optical
  Physics}\ }\textbf {\bibinfo {volume} {78}},\ \bibinfo {pages} {052331}
  (\bibinfo {year} {2008})}\BibitemShut {NoStop}%
\bibitem [{\citenamefont {Martinez}\ \emph {et~al.}(2025)\citenamefont
  {Martinez}, \citenamefont {Schnabl}, \citenamefont {del Moral}, \citenamefont
  {Dastbasteh}, \citenamefont {Crespo},\ and\ \citenamefont
  {Otxoa}}]{martinez2025leveraging}%
  \BibitemOpen
  \bibfield  {author} {\bibinfo {author} {\bibfnamefont {J.~E.}\ \bibnamefont
  {Martinez}}, \bibinfo {author} {\bibfnamefont {P.}~\bibnamefont {Schnabl}},
  \bibinfo {author} {\bibfnamefont {J.~O.}\ \bibnamefont {del Moral}}, \bibinfo
  {author} {\bibfnamefont {R.}~\bibnamefont {Dastbasteh}}, \bibinfo {author}
  {\bibfnamefont {P.~M.}\ \bibnamefont {Crespo}},\ and\ \bibinfo {author}
  {\bibfnamefont {R.~M.}\ \bibnamefont {Otxoa}},\ }\bibfield  {title} {\bibinfo
  {title} {Leveraging biased noise for more efficient quantum error correction
  at the circuit-level with two-level qubits},\ }\href@noop {} {\bibfield
  {journal} {\bibinfo  {journal} {arXiv preprint arXiv:2505.17718}\ } (\bibinfo
  {year} {2025})}\BibitemShut {NoStop}%
\bibitem [{\citenamefont {Fellous-Asiani}\ \emph {et~al.}(2025)\citenamefont
  {Fellous-Asiani}, \citenamefont {Naseri}, \citenamefont {Datta},
  \citenamefont {Streltsov},\ and\ \citenamefont
  {Oszmaniec}}]{fellous2025scalable}%
  \BibitemOpen
  \bibfield  {author} {\bibinfo {author} {\bibfnamefont {M.}~\bibnamefont
  {Fellous-Asiani}}, \bibinfo {author} {\bibfnamefont {M.}~\bibnamefont
  {Naseri}}, \bibinfo {author} {\bibfnamefont {C.}~\bibnamefont {Datta}},
  \bibinfo {author} {\bibfnamefont {A.}~\bibnamefont {Streltsov}},\ and\
  \bibinfo {author} {\bibfnamefont {M.}~\bibnamefont {Oszmaniec}},\ }\bibfield
  {title} {\bibinfo {title} {Scalable noisy quantum circuits for biased-noise
  qubits},\ }\href@noop {} {\bibfield  {journal} {\bibinfo  {journal} {npj
  Quantum Information}\ }\textbf {\bibinfo {volume} {11}},\ \bibinfo {pages}
  {145} (\bibinfo {year} {2025})}\BibitemShut {NoStop}%
\bibitem [{\citenamefont {Mirrahimi}\ \emph {et~al.}(2014)\citenamefont
  {Mirrahimi}, \citenamefont {Leghtas}, \citenamefont {Albert}, \citenamefont
  {Touzard}, \citenamefont {Schoelkopf}, \citenamefont {Jiang},\ and\
  \citenamefont {Devoret}}]{mirrahimi2014dynamically}%
  \BibitemOpen
  \bibfield  {author} {\bibinfo {author} {\bibfnamefont {M.}~\bibnamefont
  {Mirrahimi}}, \bibinfo {author} {\bibfnamefont {Z.}~\bibnamefont {Leghtas}},
  \bibinfo {author} {\bibfnamefont {V.~V.}\ \bibnamefont {Albert}}, \bibinfo
  {author} {\bibfnamefont {S.}~\bibnamefont {Touzard}}, \bibinfo {author}
  {\bibfnamefont {R.~J.}\ \bibnamefont {Schoelkopf}}, \bibinfo {author}
  {\bibfnamefont {L.}~\bibnamefont {Jiang}},\ and\ \bibinfo {author}
  {\bibfnamefont {M.~H.}\ \bibnamefont {Devoret}},\ }\bibfield  {title}
  {\bibinfo {title} {Dynamically protected cat-qubits: a new paradigm for
  universal quantum computation},\ }\href@noop {} {\bibfield  {journal}
  {\bibinfo  {journal} {New Journal of Physics}\ }\textbf {\bibinfo {volume}
  {16}},\ \bibinfo {pages} {045014} (\bibinfo {year} {2014})}\BibitemShut
  {NoStop}%
\bibitem [{\citenamefont {Lescanne}\ \emph {et~al.}(2020)\citenamefont
  {Lescanne}, \citenamefont {Villiers}, \citenamefont {Peronnin}, \citenamefont
  {Sarlette}, \citenamefont {Delbecq}, \citenamefont {Huard}, \citenamefont
  {Kontos}, \citenamefont {Mirrahimi},\ and\ \citenamefont
  {Leghtas}}]{lescanne2020exponential}%
  \BibitemOpen
  \bibfield  {author} {\bibinfo {author} {\bibfnamefont {R.}~\bibnamefont
  {Lescanne}}, \bibinfo {author} {\bibfnamefont {M.}~\bibnamefont {Villiers}},
  \bibinfo {author} {\bibfnamefont {T.}~\bibnamefont {Peronnin}}, \bibinfo
  {author} {\bibfnamefont {A.}~\bibnamefont {Sarlette}}, \bibinfo {author}
  {\bibfnamefont {M.}~\bibnamefont {Delbecq}}, \bibinfo {author} {\bibfnamefont
  {B.}~\bibnamefont {Huard}}, \bibinfo {author} {\bibfnamefont
  {T.}~\bibnamefont {Kontos}}, \bibinfo {author} {\bibfnamefont
  {M.}~\bibnamefont {Mirrahimi}},\ and\ \bibinfo {author} {\bibfnamefont
  {Z.}~\bibnamefont {Leghtas}},\ }\bibfield  {title} {\bibinfo {title}
  {Exponential suppression of bit-flips in a qubit encoded in an oscillator},\
  }\href@noop {} {\bibfield  {journal} {\bibinfo  {journal} {Nature Physics}\
  }\textbf {\bibinfo {volume} {16}},\ \bibinfo {pages} {509} (\bibinfo {year}
  {2020})}\BibitemShut {NoStop}%
\bibitem [{\citenamefont {Bonilla~Ataides}\ \emph {et~al.}(2021)\citenamefont
  {Bonilla~Ataides}, \citenamefont {Tuckett}, \citenamefont {Bartlett},
  \citenamefont {Flammia},\ and\ \citenamefont {Brown}}]{bonilla2021xzzx}%
  \BibitemOpen
  \bibfield  {author} {\bibinfo {author} {\bibfnamefont {J.~P.}\ \bibnamefont
  {Bonilla~Ataides}}, \bibinfo {author} {\bibfnamefont {D.~K.}\ \bibnamefont
  {Tuckett}}, \bibinfo {author} {\bibfnamefont {S.~D.}\ \bibnamefont
  {Bartlett}}, \bibinfo {author} {\bibfnamefont {S.~T.}\ \bibnamefont
  {Flammia}},\ and\ \bibinfo {author} {\bibfnamefont {B.~J.}\ \bibnamefont
  {Brown}},\ }\bibfield  {title} {\bibinfo {title} {The {XZZX} surface code},\
  }\href@noop {} {\bibfield  {journal} {\bibinfo  {journal} {Nature
  Communications}\ }\textbf {\bibinfo {volume} {12}},\ \bibinfo {pages} {2172}
  (\bibinfo {year} {2021})}\BibitemShut {NoStop}%
\bibitem [{\citenamefont {Roffe}\ \emph {et~al.}(2023)\citenamefont {Roffe},
  \citenamefont {Cohen}, \citenamefont {Quintavalle}, \citenamefont {Chandra},\
  and\ \citenamefont {Campbell}}]{roffe2023bias}%
  \BibitemOpen
  \bibfield  {author} {\bibinfo {author} {\bibfnamefont {J.}~\bibnamefont
  {Roffe}}, \bibinfo {author} {\bibfnamefont {L.~Z.}\ \bibnamefont {Cohen}},
  \bibinfo {author} {\bibfnamefont {A.~O.}\ \bibnamefont {Quintavalle}},
  \bibinfo {author} {\bibfnamefont {D.}~\bibnamefont {Chandra}},\ and\ \bibinfo
  {author} {\bibfnamefont {E.~T.}\ \bibnamefont {Campbell}},\ }\bibfield
  {title} {\bibinfo {title} {Bias-tailored quantum {LDPC} codes},\ }\href@noop
  {} {\bibfield  {journal} {\bibinfo  {journal} {Quantum}\ }\textbf {\bibinfo
  {volume} {7}},\ \bibinfo {pages} {1005} (\bibinfo {year} {2023})}\BibitemShut
  {NoStop}%
\bibitem [{\citenamefont {Guillaud}\ and\ \citenamefont
  {Mirrahimi}(2019)}]{guillaud2019repetition}%
  \BibitemOpen
  \bibfield  {author} {\bibinfo {author} {\bibfnamefont {J.}~\bibnamefont
  {Guillaud}}\ and\ \bibinfo {author} {\bibfnamefont {M.}~\bibnamefont
  {Mirrahimi}},\ }\bibfield  {title} {\bibinfo {title} {Repetition cat qubits
  for fault-tolerant quantum computation},\ }\href@noop {} {\bibfield
  {journal} {\bibinfo  {journal} {Physical Review X}\ }\textbf {\bibinfo
  {volume} {9}},\ \bibinfo {pages} {041053} (\bibinfo {year}
  {2019})}\BibitemShut {NoStop}%
\bibitem [{\citenamefont {Tsutsui}\ and\ \citenamefont
  {Kanno}(2024)}]{tsutsui2024bias}%
  \BibitemOpen
  \bibfield  {author} {\bibinfo {author} {\bibfnamefont {S.}~\bibnamefont
  {Tsutsui}}\ and\ \bibinfo {author} {\bibfnamefont {K.}~\bibnamefont
  {Kanno}},\ }\bibfield  {title} {\bibinfo {title} {Bias-preserving computation
  with the bit-flip code},\ }\href@noop {} {\bibfield  {journal} {\bibinfo
  {journal} {Physical Review Research}\ }\textbf {\bibinfo {volume} {6}},\
  \bibinfo {pages} {023290} (\bibinfo {year} {2024})}\BibitemShut {NoStop}%
\bibitem [{\citenamefont {Fellous-Asiani}\ \emph {et~al.}(2023)\citenamefont
  {Fellous-Asiani}, \citenamefont {Naseri}, \citenamefont {Datta},
  \citenamefont {Streltsov},\ and\ \citenamefont
  {Oszmaniec}}]{fellous2023scalable}%
  \BibitemOpen
  \bibfield  {author} {\bibinfo {author} {\bibfnamefont {M.}~\bibnamefont
  {Fellous-Asiani}}, \bibinfo {author} {\bibfnamefont {M.}~\bibnamefont
  {Naseri}}, \bibinfo {author} {\bibfnamefont {C.}~\bibnamefont {Datta}},
  \bibinfo {author} {\bibfnamefont {A.}~\bibnamefont {Streltsov}},\ and\
  \bibinfo {author} {\bibfnamefont {M.}~\bibnamefont {Oszmaniec}},\ }\bibfield
  {title} {\bibinfo {title} {Scalable noisy quantum circuits for biased-noise
  qubits},\ }\href@noop {} {\bibfield  {journal} {\bibinfo  {journal} {arXiv
  preprint arXiv:2305.02045}\ } (\bibinfo {year} {2023})}\BibitemShut {NoStop}%
\bibitem [{\citenamefont {Aharonov}(2003)}]{aharonov2003simple}%
  \BibitemOpen
  \bibfield  {author} {\bibinfo {author} {\bibfnamefont {D.}~\bibnamefont
  {Aharonov}},\ }\bibfield  {title} {\bibinfo {title} {A simple proof that
  {T}offoli and {H}adamard are quantum universal},\ }\href@noop {} {\bibfield
  {journal} {\bibinfo  {journal} {arXiv preprint quant-ph/0301040}\ } (\bibinfo
  {year} {2003})}\BibitemShut {NoStop}%
\bibitem [{\citenamefont {Nielsen}\ and\ \citenamefont
  {Chuang}(2010)}]{nielsen_and_chuang}%
  \BibitemOpen
  \bibfield  {author} {\bibinfo {author} {\bibfnamefont {M.~A.}\ \bibnamefont
  {Nielsen}}\ and\ \bibinfo {author} {\bibfnamefont {I.~L.}\ \bibnamefont
  {Chuang}},\ }\href@noop {} {\emph {\bibinfo {title} {Quantum {C}omputation
  and {Q}uantum {I}nformation}}}\ (\bibinfo  {publisher} {Cambridge University
  Press},\ \bibinfo {year} {2010})\BibitemShut {NoStop}%
\bibitem [{\citenamefont {Shepherd}\ and\ \citenamefont
  {Bremner}(2008)}]{shepherd2008instantaneous}%
  \BibitemOpen
  \bibfield  {author} {\bibinfo {author} {\bibfnamefont {D.}~\bibnamefont
  {Shepherd}}\ and\ \bibinfo {author} {\bibfnamefont {M.~J.}\ \bibnamefont
  {Bremner}},\ }\bibfield  {title} {\bibinfo {title} {Instantaneous quantum
  computation},\ }\href@noop {} {\bibfield  {journal} {\bibinfo  {journal}
  {arXiv preprint arXiv:0809.0847}\ } (\bibinfo {year} {2008})}\BibitemShut
  {NoStop}%
\bibitem [{\citenamefont {Chitambar}\ and\ \citenamefont
  {Gour}(2016)}]{chitambar2016critical}%
  \BibitemOpen
  \bibfield  {author} {\bibinfo {author} {\bibfnamefont {E.}~\bibnamefont
  {Chitambar}}\ and\ \bibinfo {author} {\bibfnamefont {G.}~\bibnamefont
  {Gour}},\ }\bibfield  {title} {\bibinfo {title} {Critical examination of
  incoherent operations and a physically consistent resource theory of quantum
  coherence},\ }\href@noop {} {\bibfield  {journal} {\bibinfo  {journal}
  {Physical Review Letters}\ }\textbf {\bibinfo {volume} {117}},\ \bibinfo
  {pages} {030401} (\bibinfo {year} {2016})}\BibitemShut {NoStop}%
\bibitem [{\citenamefont {Streltsov}\ \emph {et~al.}(2017)\citenamefont
  {Streltsov}, \citenamefont {Adesso},\ and\ \citenamefont
  {Plenio}}]{streltsov2017colloquium}%
  \BibitemOpen
  \bibfield  {author} {\bibinfo {author} {\bibfnamefont {A.}~\bibnamefont
  {Streltsov}}, \bibinfo {author} {\bibfnamefont {G.}~\bibnamefont {Adesso}},\
  and\ \bibinfo {author} {\bibfnamefont {M.~B.}\ \bibnamefont {Plenio}},\
  }\bibfield  {title} {\bibinfo {title} {Colloquium: {Q}uantum coherence as a
  resource},\ }\href@noop {} {\bibfield  {journal} {\bibinfo  {journal}
  {Reviews of Modern Physics}\ }\textbf {\bibinfo {volume} {89}},\ \bibinfo
  {pages} {041003} (\bibinfo {year} {2017})}\BibitemShut {NoStop}%
\bibitem [{\citenamefont {Jones}\ \emph {et~al.}(2024)\citenamefont {Jones},
  \citenamefont {Linden},\ and\ \citenamefont
  {Skrzypczyk}}]{jones2024hadamard}%
  \BibitemOpen
  \bibfield  {author} {\bibinfo {author} {\bibfnamefont {B.~D.}\ \bibnamefont
  {Jones}}, \bibinfo {author} {\bibfnamefont {N.}~\bibnamefont {Linden}},\ and\
  \bibinfo {author} {\bibfnamefont {P.}~\bibnamefont {Skrzypczyk}},\ }\bibfield
   {title} {\bibinfo {title} {The {H}adamard gate cannot be replaced by a
  resource state in universal quantum computation},\ }\href@noop {} {\bibfield
  {journal} {\bibinfo  {journal} {Quantum}\ }\textbf {\bibinfo {volume} {8}},\
  \bibinfo {pages} {1470} (\bibinfo {year} {2024})}\BibitemShut {NoStop}%
\bibitem [{\citenamefont {Calderbank}\ and\ \citenamefont
  {Shor}(1996)}]{calderbank1996good}%
  \BibitemOpen
  \bibfield  {author} {\bibinfo {author} {\bibfnamefont {A.~R.}\ \bibnamefont
  {Calderbank}}\ and\ \bibinfo {author} {\bibfnamefont {P.~W.}\ \bibnamefont
  {Shor}},\ }\bibfield  {title} {\bibinfo {title} {Good quantum
  error-correcting codes exist},\ }\href@noop {} {\bibfield  {journal}
  {\bibinfo  {journal} {Physical Review A}\ }\textbf {\bibinfo {volume} {54}},\
  \bibinfo {pages} {1098} (\bibinfo {year} {1996})}\BibitemShut {NoStop}%
\bibitem [{\citenamefont {Steane}(1996)}]{steane1996error}%
  \BibitemOpen
  \bibfield  {author} {\bibinfo {author} {\bibfnamefont {A.~M.}\ \bibnamefont
  {Steane}},\ }\bibfield  {title} {\bibinfo {title} {Error correcting codes in
  quantum theory},\ }\href@noop {} {\bibfield  {journal} {\bibinfo  {journal}
  {Physical Review Letters}\ }\textbf {\bibinfo {volume} {77}},\ \bibinfo
  {pages} {793} (\bibinfo {year} {1996})}\BibitemShut {NoStop}%
\bibitem [{\citenamefont {Gottesman}(1997)}]{gottesman1997stabilizer}%
  \BibitemOpen
  \bibfield  {author} {\bibinfo {author} {\bibfnamefont {D.}~\bibnamefont
  {Gottesman}},\ }\emph {\bibinfo {title} {Stabilizer Codes and Quantum Error
  Correction}},\ \href@noop {} {Ph.D. thesis},\ \bibinfo  {school} {California
  Institute of Technology} (\bibinfo {year} {1997})\BibitemShut {NoStop}%
\end{thebibliography}%

\end{document}